\documentclass{amsart}
\usepackage{amsmath, amssymb,amsthm}
\usepackage{graphicx}
\usepackage{txfonts}
\usepackage{bbm}
\usepackage{enumerate}

\usepackage{graphics}
\usepackage{float}

\theoremstyle{definition}
\newtheorem{defn}{Definition}[section]

\theoremstyle{plain}
\newtheorem{thm}{Theorem}

\newtheorem{lem}[defn]{Lemma}
\newtheorem{cor}[defn]{Corollary}
\newtheorem{prop}[defn]{Proposition}
\theoremstyle{remark}

\theoremstyle{dotless}

 \newcommand{\eps}{\varepsilon}

 \newcommand{\M}{\mathcal{M}}

\newcommand{\E}{\mathbb{E}}

\newcommand{\N}{\mathbb{N}}
\newcommand{\Q}{\mathbb{Q}}
\newcommand{\R}{\mathbb{R}}
\newcommand{\m}{\Q}

\newcommand{\law}{\mathrm{Law}}
\usepackage{bbm}
\newcommand{\1}{\mathbbm 1 }

\newcommand{\no}{\noindent}
\newcommand{\be}{\begin{equation}}
\newcommand{\ee}{\end{equation}}

\newcommand{\bea}{\begin{eqnarray}}
\newcommand{\bes}{\begin{subequations}}
\newcommand{\ees}{\end{subequations}}
\newcommand{\bgt}{\begin{gather}}
\newcommand{\egt}{\begin{gather}}
\newcommand{\eea}{\end{eqnarray}}
\newcommand{\beaa}{\begin{eqnarray*}}
\newcommand{\eeaa}{\end{eqnarray*}}

\newcommand{\cal}{\mathcal}

\renewcommand{\SS}{\mathcal S}

\newcommand{\dm}{\tfrac{d\mu_1}{d\lambda}}
\newcommand{\dn}{\tfrac{d\mu_2}{d\lambda}}

\def\jelname{{\bfseries JEL Classification}\enspace}
\def\jel#1{\par\addvspace\medskipamount{\rightskip=0pt plus1cm
\def\and{\ifhmode\unskip\nobreak\fi\ $\cdot$
}\noindent\jelname\ignorespaces#1\par}}
\def\subclassname{{\bfseries Mathematics Subject Classification
(2010)}\enspace}
\def\subclass#1{\par\addvspace\medskipamount{\rightskip=0pt plus1cm
\def\and{\ifhmode\unskip\nobreak\fi\ $\cdot$
}\noindent\subclassname\ignorespaces#1\par}}
\begin{document}
\title[Model-independent Bounds by Mass Transport]{Model-independent Bounds for Option Prices: \\ A Mass Transport Approach}
\author{Mathias Beiglb\"ock}
\author{Pierre Henry-Labord\`ere}
\author{Friedrich Penkner}
\thanks{ 
\hspace{-0.5cm}M.\ Beiglb\"ock, University of Vienna, Department of Mathematics, {\tt mathias.beiglboeck@univie.ac.at} \\
P.\ Henry-Labord\`ere, Global Markets Quantitative Research, Soci\'{e}t\'{e} G\'{e}n\'{e}rale, {\tt pierre.henry-labordere@sgcib.com}\\
F.\ Penkner, University of Vienna, Department of Mathematics, {\tt friedrich.penkner@univie.ac.at}\\
The first author thanks the FWF for support through project p21209.}


\maketitle

\begin{abstract}\vspace{-0.5cm}
In this paper we investigate  model-independent bounds for exotic options 
written on a risky asset using infinite-dimensional linear
programming methods.
Based on arguments from the theory of Monge-Kantorovich mass-transport
we establish a dual version of the problem that has a natural
financial interpretation in terms of semi-static hedging.
In particular we prove that there is no duality gap.
\medskip

 \subclass{91G20, 91G80}
 \jel{C61, G13}
\end{abstract}

\section{Introduction}\label{sec:intro}
Since the introduction of the Black-Scholes paradigm, several
alternative models which allow to capture the risk of exotic
options have emerged: stochastic volatility
models,  local volatility models, jump-diffusion models, mixed local stochastic volatility
models. These models depend on various parameters which can be
calibrated more or less  accurately to market prices of liquid
options (such as vanilla options). 
This calibration procedure does not  uniquely set the dynamics of
forward prices which are only required  to be (local) martingales
according to the no-arbitrage framework. This could lead to a wide
range of prices of a given exotic option when evaluated using
different models calibrated to the same market data.

 In practice, it would be interesting to know lower and upper
bounds for exotic options produced by models  calibrated to the same
market data, and therefore with similar marginals. If bounds are
tight enough, they would be used to detect the possibility of arbitrage in market
prices, provided these bounds have an interpretation as investment
strategies. This problem has already been studied in the case of
exotic options written on multi-assets $(S^1, \dotsc , S^k)$
observed at the same time $T$ \cite{BePo02,ChDeDhVa08,HoLaWa05a,HoLaWa05b,LaWa05,LaWa04}.  
Within the class of models with
fixed marginals $\left(\law({S^1_T}), \dotsc, \law({S^k_T})\right)$
at $T$, the search for lower/upper bounds involves
infinite-dimensional linear programming issues. Analytical
expressions have been obtained in the case of basket options
\cite{LaWa05,LaWa04}. These correspond to
the determination of optimal copulas.  In practice, these bounds are not tight as the information of marginals is not restrictive enough.

Here we focus on multi-period models and general path-dependent options. 
This problem  is more
involved as we have to impose that the asset price $S_{t}$ is a
discrete time martingale\footnote{For the sake of simplicity, we assume
zero interest rate and no cash/yield dividends. This assumption can
be relaxed by considering the process $f_t$  introduced in
\cite{He09} (see equation 14)  which has the property to be a local
martingale.} satisfying marginal restrictions. We review the existing literature on the subject in Section \ref{sec:intro:subsec:comparison} below.

In our setting the problem of determining the interval of consistent prices of a
given exotic option can be cast as a (primal) infinite-dimensional
linear programming problem. We propose a dual problem that has a
practically relevant interpretation in terms of trading strategies
and prove that there is no duality gap under rather mild regularity
assumptions.
\pagebreak
\subsection{Setting}\label{sec:intro:subsec:setting}

  In the following, we fix an exotic option depending only on
the value of a single asset $S$ at discrete times $t_1 < \ldots <
t_n$
 and denote by $\Phi(S_1,\ldots, S_n)$ its
payoff, where we suppose $\Phi$ to be some measurable function. In the no-arbitrage framework, the standard approach is to
postulate a model, that is, a probability measure $\Q$ on $\R^n$
under which the coordinate process $(S_i)_{i=1}^n$
$$S_i:\R^n \to \R,\ S_i(s_1,\ldots, s_n)=s_i,\ i=1,\ldots,n,$$
 is required to be a (discrete)
martingale in its own filtration. By $S_0=s_0$ we denote the current spot price.  The fair value of $\Phi$ is then given as
the expectation of the payoff
$$\E_\Q [\Phi]. $$
 Additionally, we impose that our
model is calibrated to a continuum of call options with payoffs
$\Phi_{i,K}(S_i)=(S_i-K)^+, K\in \R$ at each date $t_i$ and price
    \begin{align} 
	{\cal C}(t_i,K)= \E_\Q[ \Phi_{i,K}] = \int_{\R^+}
	(s-K)^+\, d\law_{S_i}(s). \label{vanilla} 
    \end{align}
 Plainly
\eqref{vanilla} is tantamount to prescribing probability measures
$\mu_1,  \ldots, \mu_n$ on the real line\footnote{The cumulative
distribution function of $\mu_i$ can be read off the call prices
through $F_i(K)= 1 - \lim_{\eps\downarrow 0}  1/\eps \big[ {\cal
C}(t_i,K)-{\cal C}(t_i,K+\varepsilon) \big]$ for $i=1,\ldots,n$.
Concerning the mathematical finance application it would be sufficient to consider strikes $K\geq 0$ and 
marginals which are concentrated on the positive half-line. 
We prefer to go with the more general case since the proofs are not more complicated. A technical difference is that call prices satisfy only $\lim_{K\to -\infty}{\cal C}(t_i,K) - K = s_0$ rather than the simpler ${\cal C}(t_i,0)  = s_0$ in the case where $S$ is assumed to be non-negative.} such that the \emph{one dimensional marginals} of $\Q$ satisfy
    $$
	\Q^i=\law_{S_i}=\mu_i\,  \; \mbox{for all}  \; i=1,\ldots, n.
    $$

\subsection{Primal formulation}\label{sec:intro:subsec:primal}

For further reference, we denote by $\M(\mu_1, \ldots, \mu_n)$ the
set of all martingale measures $\Q$ on  (the pathspace) $\R^n$ having marginals
$\Q^1=\mu_1, \ldots, \Q^n=\mu_n$ and mean $s_0$. Equivalently, we have
$\Q\in\M(\mu_1, \ldots, \mu_n)$ if and only if
$\E_\Q[S_i|S_1,\ldots,S_{i-1}]=S_{i-1}$ for $i=2, \ldots, n$ and
$\E_\Q[ \Phi_{i,K}]=   {\cal C}(t_i,K)$ for all $K\in \R$ and $i=1,
\ldots, n$.

  Following the tradition customary in the optimal transport literature we concentrate on the lower
  bound and consider the \emph{primal problem} \begin{align}
P=\inf\big\{ \E_\Q[ \Phi]: \Q \in \M(\mu_1, \ldots, \mu_n)\big\}.
\label{primal}
\end{align}

\subsection{Dual formulation}\label{sec:intro:subsec:dual}
The dual formulation corresponds to the construction of a
\emph{semi-static subhedging strategy} 
  consisting of the sum of a
static vanilla portfolio and a delta strategy.\footnote{Similar strategies are considered in \cite{DaHo07,Co07} where they are used to subreplicate a European option based on finitely many given call options.} 
 More precisely, we
are interested in payoffs of the form
    \begin{align}
	\Psi_{(u_i), (\Delta_j)}(s_1, \ldots, s_n)= \sum_{i=1}^n u_i(s_i)+
	\sum_{j=1}^{n-1} \Delta_j(s_1, \ldots, s_{j})(s_{j+1}-s_{j} ),\,
	\quad s_1, \ldots, s_n\in \R, \label{DualCandidate}
    \end{align}
where the functions $u_i:\R\to \R$ are $\mu_i$-integrable ($i=1,
\ldots, n$)  and the functions $\Delta_j:\R^j\to \R$ are assumed to
be bounded measurable ($j=1,\ldots, (n-1)$).\footnote{It might be expected that the delta strategy in \eqref{DualCandidate} should also include a constant $\Delta_0$ multiplier of $(s_1-s_0)$ corresponding to an initial forward position. However this term is not necessary as it can be subsumed into the term $u_1$.}

If these functions lead to a strategy which is subhedging in the
sense
$$\Phi\geq \Psi_{(u_i), (\Delta_j)}$$
we have for every pricing measure $\Q\in \M(\mu_1, \ldots, \mu_n)$
the obvious inequality
    \begin{align}\label{PriceRep}
	\E_\Q [\Phi]\geq \E_\Q[\Psi_{(u_i), (\Delta_j)}]=
	\E_\Q\Big[ \sum_{i=1}^n u_i(S_i)\Big]  =
	\sum_{i=1}^n \E_{\mu_i} [u_i].
    \end{align}
This leads us to consider the  \emph{dual problem}
    \begin{align}
	D= & \sup\Big\{ \sum_{i=1}^n \E_{\mu_i} [u_i] : \exists\, \Delta_1,
	\ldots, \Delta_{n-1} \mbox{ s.t. }\Psi_{(u_i), (\Delta_j)}\leq
	\Phi\Big\};\label{Dual}
    \end{align}
which, by \eqref{PriceRep}, satisfies
    \begin{align}\label{TrivialDuality}
	P\geq D. 
    \end{align}

\subsection{Semi-static subhedging}\label{sec:intro:subsec:semi_static}
The dual formulation corresponds to the construction of a semi-static subhedging
portfolio consisting of static vanilla options  $ u_i(S_i)$ and investments in the risky asset according to
the self-financing trading strategy $\big(\Delta_j(S_1, \ldots, S_{j})\big)_{j=1}^{n-1}.$

We note the financial interpretation of inequality
\eqref{TrivialDuality}: suppose somebody offers the option $\Phi$ at
a price $p< D$. Then there exists ${(u_i), (\Delta_j)}$ with
$\Psi_{(u_i), (\Delta_j)} \leq \Phi$ with price $\sum_{i=1}^n
\E_{\mu_i}[u_i]$ strictly larger than $p$. Buying $\Phi $ and going
short in $\Psi_{(u_i), (\Delta_j)}$, the arbitrage can be locked in.

The crucial question is of course if \eqref{TrivialDuality} is
sharp, i.e.\ if every option priced below $P$ allows for an
arbitrage by means of semi-static subhedging. In Theorem \ref{MainTheorem} below we show that this is
the case under relatively mild assumptions.

\medskip

Of course it is a classical  theme of Mathematical Finance that the extremal martingale prices of a financial derivative correspond to the minimal or maximal initial capital necessary for sub-/super-replication, respectively. This is precisely the replication theorem of mathematical finance, which is a corollary of the fundamental theorem of asset pricing. The novelty of our contribution is that we establish a robust, \emph{model-free} version of this result.

\subsection{Main result}\label{sec:intro:subsec:main_result}

\begin{thm}\label{MainTheorem}
Assume that $\mu_1,\ldots, \mu_n$ are Borel probability measures on
$\R$ such that $\M(\mu_1,\ldots, \mu_n) $ is non-empty. Let
$\Phi:\R^n\to (-\infty, \infty]$ be a lower semi-continuous function such that
    \begin{align}\label{PayoffLB1}
	\Phi(s_1, \ldots , s_n)\geq - K\cdot  (1+|s_1|+\ldots+|s_n|)
    \end{align}
on $\R^n$ for some constant $K$. Then there is no duality gap, i.e.\
$P=D$. Moreover, the primal value 
 $P$ is attained, i.e.\ there exists a martingale measure $\m\in\M(\mu_1, \ldots, \mu_n)$ such that $P=\E_\Q[\Phi]$.  
 
The dual supremum is in general not attained (cf.~Proposition~\ref{DualEx} below).
\end{thm}
Our approach to this result is based on the duality theory of
optimal transport which is briefly introduced in Section 2; the
actual proof  will be given in Section 3 with the help of the
Min-Max Theorem of decision theory.
We conclude this introductory section by a short discussion of the
content of Theorem \ref{MainTheorem}.

 \medskip

The assumption $\M(\mu_1,\ldots, \mu_n)\neq \emptyset $ excludes the
degenerate case in which no calibrated market model exists. For the
existence of a martingale measure having marginals $\mu_1, \ldots,
\mu_n$ it is necessary and sufficient that these measures possess
the same finite first moments and increase in the \emph{convex
order}, i.e.\ $\E_{\mu_1}\phi\leq\ldots\leq \E_{\mu_n}\phi$ for each
convex function $\phi:\R\to\R$ (cf.\ \cite{St65}).\footnote{In more
financial terms this means that ${\cal C}(t,K)$ is increasing in $t$
for each fixed $K\in \R$.}

\medskip

Having the financial interpretation in mind, it is important that the value $D$ of the dual problem remains unchanged if a smaller set of subhedging strategies $\Psi_{(u_i),(\Delta_j)}$ is used. 
In the proof of Theorem \ref{MainTheorem} we show that it is sufficient to  consider functions $u_1, \ldots, u_n$  which
are linear combinations of finitely many call options  (plus one
position in the bond resp.\ the stock); at the same time  $\Delta_1,
\ldots, \Delta_{n-1}$ can be taken to be continuous and bounded. 
This means that for every $\eps>0$ there exist $b, c_{i,l}, K_{i,l}\in\R,\, i=1, \ldots,n,\, l=1, \ldots,m_i,\, \Delta_j\in \mathcal {C}_b(\R^j)$, $j=0, \ldots, n-1$ such that
\begin{align}
b +\sum_{i=1}^n \sum_{l=1}^{m_i}c_{i,l} (s_i-K_{i,l})_+ +\sum_{j=0}^{n-1} \Delta_{j}(s_1, \ldots, s_{j})(s_{j+1}-s_{j})\leq\Phi(s_1,\ldots, s_n),
\end{align}
and the corresponding price 
\begin{align}
p=b +\sum_{i=1}^n \sum_{l=1}^{m_i}c_{i,l} {\cal C}(t_i,K_{i,l}) 
\end{align}
is $\eps$-close to the primal value $P$.

\medskip

Condition \eqref{PayoffLB1} could be somewhat relaxed. For instance
it is sufficient to demand that the function $\Phi$ is bounded from
below by a sum of integrable functions. However, in this case it
is necessary to allow for dual strategies that use European options
beyond call options and we will not pursue this further.

\medskip

We conclude this introductory section by noting that an \emph{upper} bound for the price of the option $\Phi$ can be given by means of \emph{semi-static superhedging}. Applying Theorem \ref{MainTheorem} to the function $-\Phi$ we obtain that this bound is sharp: 
\begin{cor}\label{MainCorollary}
Assume that $\mu_1,\ldots, \mu_n$ are Borel probability measures on
$\R$ such that $\M(\mu_1,\ldots, \mu_n) $ is non-empty. Let
$\Phi:\R^n\to [-\infty, \infty)$ be an upper semi-continuous function such that
\begin{align}\label{PayoffLB2}
\Phi(s_1, \ldots , s_n)\leq K\cdot  (1+|s_1|+\ldots+|s_n|)
\end{align}
on $\R^n$ for some constant $K$. Then there is no duality gap
\begin{align}
P =& \sup\Big\{\E_\m\Phi: \m\in\M(\mu_1, \ldots, \mu_n)\Big\} \\
  =& \inf\Big\{ \sum_{i=1}^n \E_{\mu_i} [u_i] : \exists\, \Delta_1,
\ldots, \Delta_{n-1} \mbox{ s.t. }\Psi_{(u_i), (\Delta_j)}\geq
\Phi\Big\} = D. 
\end{align}

The supremum is attained, i.e.\ there exists a maximizing martingale measure. 

\end{cor}

\medskip

\subsection{Comparison with previous results}\label{sec:intro:subsec:comparison}

The main novelty of our approach is that we apply the theory of optimal transport in mathematical finance, more specifically, to obtain robust model-independent bounds on option prices. 
A time-continuous analysis of the present connection between optimal transport and mathematical finance is contained in the parallel work to the present one by Galichon, Henry-Labord\`ere, Touzi \cite {GaHeTo11} (see also  \cite{HeObSpTo12}) where a stochastic control approach is used.  

We point out that the problem of model independent pricing is classically approached in the literature by means of the Skorokhod embedding problem, see the informative survey paper by Hobson \cite{Ho11}.   Also the notion of semi-static hedges is well-established (see for instance \cite[Section 2.6]{Ho11}). 

The problem of robust pricing in a multi-period setting has previously been studied in the case of specific exotic options.  Hodges and Neuberger \cite{HoNe00} are mainly interested in the case of Barrier options. Albrecher,  Mayer and Schoutens  produce an explicit bound (based on conditioning arguments) in the case of an Asian option in  discrete time  and give a feasible subreplicating strategy associated to it \cite{AlMaSc08}. The problem to explicitly give the optimal lower/upper bounds  seems harder and remains open to the best of our knowledge. A numerical implementation of our dual approach in the Asian option setting is given in \cite{He11a}. 

In the continuous-time setting, the problem has been treated for instance in the case of lookback options \cite{Ho98a}, variance/volatility options \cite{CaLe10,CoWa11,HoKl12} and double-(no) touch options \cite{CoOb11a,CoOb11b}. These solutions are mainly based on Skorokhod-stopping techniques and differ from our approach also in that only the marginal at the maturity is incorporated. Extensions to the multi-marginal case are addressed in \cite{BrHoRo01b,MaYo02,HoPe02,HeObSpTo12}. 

A result similar to our findings  was recently proved for forward-start options by Hobson and Neuberger \cite{HoNe12}. In the terminology of Corollary \ref{MainCorollary} they show that $P=D$ in the case where $n=2$ and the payoff function is given by $$ \Phi(s_1,s_2)= |s_2-s_1|.$$
In contrast to our paper, their approach is more constructive and they obtain maximizers for the dual problem in particular cases. Here some care is needed in certain (pathological) situations, see Proposition \ref{DualEx} below.

\section{Optimal Transport}\label{sec:opt_trans}

In the usual theory of Monge-Kantorovich optimal transport\footnote{See
\cite{Vi03,Vi09} for an extensive account on the theory of optimal
transportation.} one
considers two probability spaces $(X_1,\mu_1)$, $(X_2,\mu_2)$ and
the problem is to find a ``cheap'' way of transporting $\mu_1$ to
$\mu_2$. Following Kantorovich, a transport plan is formalized as
probability measure $\pi$ on $X_1\times X_2$  which has
$X_1$-marginal $\mu_1$ and $X_2$-marginal $\mu_2$. 

We will come back to the two dimensional case in Section 4 below;
for now we turn to the \emph{multidimensional version} of the transport problem which will be the main tool in our proof of Theorem \ref{MainTheorem}. Subsequently we consider probability
measures $\mu_1, \ldots, \mu_n$ on the real line\footnote{Most of
the basic results are equally true for polish probability spaces
$(X_1,\mu_1),\ldots, (X_n, \mu_n)$, but we do not need this
generality here.} which have finite first moments. The set
$\Pi(\mu_1, \ldots, \mu_n)$ of \emph{transport plans} consists of
all Borel probability measures on $\R^n$  with marginals $\mu_1,
\ldots, \mu_n$. A \emph{cost function} is a measurable function
$\Phi:\R^n\to (-\infty, \infty]$ which is bounded from below in the
sense that there exist $\mu_i$-integrable functions $u_i$,
$i=1,\ldots, n$  such that
    \begin{align}\label{BoundedBelow}
	\Phi\geq u_1\oplus\ldots\oplus u_n,
    \end{align}
where $ u_1\oplus\ldots\oplus u_n(x_1,\ldots, x_n):=u_1(x_1)+
\ldots+u_n(x_n)$. Given a cost function $\Phi $ and a transport plan
$\pi$ the \emph{cost functional} is defined as
    \begin{align}\label{KantorovichPrimal}
	\textstyle{ I_\pi(\Phi)= \int_{\R^n} \Phi\, d\pi}\, .
    \end{align}
Note that this integral is well defined (assuming possibly the value
$+\infty$) by \eqref{BoundedBelow}. The \emph{primal
Monge-Kantorovich problem} is then to minimize $I_{\pi}(\Phi)$ over
the set of all transport plans $\pi\in\Pi(\mu_1, \ldots, \mu_n)$.

\medskip

Given $\mu_i$-integrable functions $u_i$, $i=1,\ldots, n$, such that
    \begin{align}\label{DualAdmissible}
	\Phi\geq u_1\oplus\ldots\oplus u_n,
    \end{align}
we have for every transport plan $\pi$
    \begin{align}\label{TMKD} \textstyle{
    \int \Phi\,d\pi\geq \int u_1\oplus\ldots\oplus u_n\, d\pi= \int
    u_1\,d\mu_1+\ldots+\int u_n\,d\mu_n.}
    \end{align}
The \emph{dual} part of the Monge-Kantorovich problem is to maximize
the right hand side of \eqref{TMKD} over a suitable class of functions
satisfying \eqref{DualAdmissible}.

Starting already with Kantorovich, there has been a long line of research on the question in which setting the optimal values of primal and dual problem agree, we refer the reader to \cite[p.\ 88f.]{Vi09} for an account of the history of the problem.
 For our intended application, we need to restrict the dual
maximizers to functions  in
    $$
    \SS= \Big\{u:\R\to \R: u(x)= a +b x + \sum_{i=1}^m c_i(x-k_i)_{+},\,  a,b,c_i,k_i\in\R\Big\},
    $$
i.e., we will employ the following Monge-Kantorovich duality
theorem.
\begin{prop}\label{MKDuality} Let $\Phi: \R^n\to (-\infty, \infty]$ be a lower semi-continuous function satisfying
    \begin{align}\label{PayoffLB}
	\Phi(s_1, \ldots , s_n)\geq - K\cdot  (1+|s_1|+\ldots+|s_n|)
    \end{align}
on $\R^n$ for some constant $K$ and let $\mu_1, \ldots, \mu_n$ be 
probability measures on $\R$ having finite first moments. Then
    \begin{align*}
    P_{MK}(\Phi)=& \inf\{I_{\pi}(\Phi): \pi\in \Pi(\mu_1, \ldots, \mu_n)\} \\
		=& \sup\Big\{ \sum_{i=1}^n \int u_i\, d\mu_i:u_1\oplus\ldots\oplus
    u_n\leq \Phi, \, u_i\in \SS\Big\}=D_{MK}(\Phi)\,.
    \end{align*}
\end{prop}
The dual bound $D_{MK}$ could be realized by holding a static position in European options with respective maturity date $t_i$ and payoff $u_i$. 
This static portfolio 
with intrinsic value $\sum_{i=1}^n u_i$ and market value $ \sum_{i=1}^n \E_{\mu_i}[ u_i]$ subreplicates the payoff $\Phi$ at maturity.

\medskip

We postpone the proof of Proposition \ref{MKDuality} to the Appendix and continue with our discussion.

The set of transport plans $\Pi(\mu_1, \ldots, \mu_n)$ carries a
natural topological structure: it is a compact convex subset of the
space of finite (signed) Borel measures  equipped with the weak
topology induced by the bounded continuous functions $C_b(\R^n)$.
(Compactness of $\Pi(\mu_1,\ldots, \mu_n) $ is essentially  a
consequence of Prokhorov's theorem, for a proof we refer the reader
to \cite[Lemma 4.4]{Vi09}.)

Subsequently we want to study the set of transport plans which are
also martingales. Therefore we will assume from now on that the
measures $\mu_1, \ldots,\mu_n$ are increasing in the convex order such that
$\M(\mu_1, \ldots, \mu_n)$ is a non-empty subset of $\Pi(\mu_1,
\ldots, \mu_n)$. It will be crucial for our purposes that also
$\M(\mu_1, \ldots, \mu_n)$ is compact in the weak topology. To
establish this  we need two auxiliary lemmas.

\begin{lem}\label{WeakExtended}
Let $c:\R^n\to \R$ be continuous and  assume  that there exists  a
constant $K$ such that
$$ |c(x_1, \ldots, x_n)| \leq K(1+|x_1|+\ldots+|x_n|)$$
for all $x_1\in X_1, \ldots, x_n\in X_n$. Then the mapping
$$\pi \mapsto \int_{\R^n} c\, d\pi$$
is continuous on $\Pi(\mu_1, \ldots, \mu_n).$
\end{lem}
\begin{proof}
Since we assume that $\mu_1, \ldots, \mu_n$ have finite first
moments, $\int_{\R^n\setminus [-a,a]^n} c\, d\pi $ converges to $0$
uniformly in $\pi\in\Pi(\mu_1, \ldots, \mu_n)$ as $a\to \infty$.
 \end{proof}

\begin{lem}\label{DifferentChar}
Let $\pi\in \Pi(\mu_1, \ldots, \mu_n)$. Then the following are
equivalent.
\begin{enumerate}
\item $\pi \in \M(\mu_1, \ldots, \mu_n)$.
\item  For $1\leq j\leq n-1$ and for every continuous bounded function $\Delta:\R^{j}\to \R$ we have
$$ \int_{ \R^n}\Delta(x_1,\ldots ,x_{j}) (x_{j+1}-x_{j})\,d\pi(x_1,\ldots, x_{n}) = 0.$$
\end{enumerate}
\end{lem}
\begin{proof}
Plainly, (1)  asserts that whenever $A\subseteq R^j$, $j=1,\ldots,
(n-1)$ is Borel measurable, then
$$\int_{ \R^n}I_A(x_1,\ldots ,x_{j}) (x_{j+1}-x_{j})\,d\pi(x_1,\ldots, x_{n}) = 0.$$
Using standard approximation techniques one obtains that this is
equivalent to (2).
 \end{proof}

\begin{prop}
The set $\M(\mu_1, \ldots, \mu_n)$ is compact in the weak topology.
\end{prop}
\begin{proof}
Since $\M(\mu_1, \ldots, \mu_n)$ is contained in the compact set
$\Pi(\mu_1, \ldots, \mu_n)$ it is sufficient to prove that it is
closed. By Lemma \ref{DifferentChar}, $ \M(\mu_1, \ldots, \mu_n)$ is
the intersection of the sets
\begin{align}\label{ThePreimage}
\Big\{\pi \in \Pi( \mu_1, \ldots, \mu_n): \int_{ \R^n}f(x_1,\ldots
,x_{j}) (x_{j+1}-x_{j})\,d\pi(x_1,\ldots, x_{n}) = 0\Big\},
\end{align}
where $j =1,\ldots, n-1$ and $f:\R^j\to \R$ runs through all
continuous bounded functions.  By Lemma \ref{WeakExtended}
the sets in \eqref{ThePreimage} are closed.
 \end{proof}

 \medskip

\section{Proof of Theorem \ref{MainTheorem}}\label{sec:proof_of_main_thm}

Our argument combines a Monge-Kantorovich duality theorem (in the
form of Proposition \ref{MKDuality}) with  the following Min-Max theorem of
decision theory which we cite here from \cite[Thm. 45.8]{St85} (another reference is \cite[Thm. 2.4.1]{AdHe96}).
\begin{thm}\label{MinMax}
Let $K,T$ be convex subsets of vector spaces $V_1$ resp.\ $V_2$,
where $V_1$ is locally convex  and let $f:K\times T\to \R$. If
\begin{enumerate}
\item $K$ is compact,
\item $f(.,y)$ is continuous and convex on $K$ for every $y\in T$,
\item $f(x,.)$ is concave on $T$ for every $x\in K$
\end{enumerate}
then
\begin{equation*}
 \sup_{y\in T} \inf_{x\in K} f(x,y)= \inf_{x\in K}\sup_{y\in T} f(x,y).
\end{equation*}

\end{thm}

\begin{proof}[of Theorem \ref{MainTheorem}.]

As we want to show that the subhedging portfolios can be formed
using just call options, we will restrict ourselves to dual
candidates $\Psi_{(u_i),(\Delta_j)}$ satisfying $u_i\in\SS,
i=1,\ldots, n$ (and $\Delta_j\in C_b(\R^j), j=1,\ldots, n-1$).

If the assertion of Theorem \ref{MainTheorem} holds true for a
function $\Phi$ and if $u_1, \ldots, u_n\in \SS$ then the assertion
carries over to $\Phi'= \Phi + u_1\oplus\ldots\oplus u_n.$ Therefore
we may assume without loss of generality that $\Phi\geq 0$.

Moreover for now we make the additional assumption that $\Phi\in C_b(\R^n)$; we will get rid of this extra condition later.

We will apply Theorem \ref{MinMax} to the compact convex set
$K=\Pi(\mu_1, \ldots, \mu_n)$, the convex set $T = C_b(\R)\times
\ldots \times C_b(\R^{n-1})$ of $(n-1)$-tuples of continuous bounded
functions on $\R^j, j=1,\ldots, (n-1)$ and the function
\begin{equation}
 f(\pi,(\Delta_j))=\int \Phi(x_1, \dotsc, x_n) - \sum_{j = 1}^{n-1}\Delta_j(x_1, \dotsc, x_{j})(x_{j+1} - x_{j})\, d\pi(x_1, \dotsc , x_n).
\end{equation}
Clearly the assumptions of Theorem \ref{MinMax} are satisfied, the
continuity of $f(.,(\Delta_j))$ on $\Pi(\mu_1, \ldots, \mu_n)$ being
a consequence of Lemma \ref{WeakExtended}.

We then find
\begin{align}
\label{nr11}
D& \geq \sup_{u_i\in\SS,\, \Delta_j\in C_b(\R^{j}),\, \Psi_{(u_i),(\Delta_j)}\leq \Phi} \sum_{i = 1}^{n} \int u_i\, d\mu_i \\
  \label{nr21}
&= \sup_{\Delta_j\in C_b(\R^{j})}\  \sup_{u_i\in\SS,\, \sum_{i = 1}^{n} u_i(x_i) \leq \Phi(x_1, \dotsc , x_n) - \sum_{j = 1}^{n - 1} \Delta_j(x_1, \dotsc, x_j)(x_{j+1} - x_j) } \sum_{i = 1}^{n} \int u_i\, d\mu_i \\
\label{nr31}
&= \sup_{\Delta_j\in C_b(\R^{j})}\  \inf_{\pi\in\Pi(\mu_1,\dotsc,\mu_n)}\int \Phi(x_1,\dotsc ,x_n)- \sum_{j = 1}^{n-1}\Delta_j(x_1,\dotsc,x_j)(x_{j+1}-x_j)\, d\pi\\
\label{nr41}
&= \inf_{\pi\in\Pi(\mu_1,\dotsc,\mu_n)}\ \sup_{\Delta_j\in C_b(\R^{j})} \int \Phi(x_1,\dotsc ,x_n)- \sum_{j = 1}^{n-1}\Delta_j(x_1,\dotsc,x_j)(x_{j+1}-x_j)\, d\pi\\
\label{nr51} &=\inf_{\m\in\M(\mu_1,\dotsc,\mu_n)} \int
\Phi(x_1,\dotsc ,x_n)\, d\Q = P.
\end{align}
Here Proposition \ref{MKDuality} is applied to  
\mbox{$\Phi(x_1,\dotsc ,x_n)- \sum_{j =
1}^{n-1}\Delta_j(x_1,\dotsc,x_j)(x_{j+1}-x_j)$}  to establish the
equality between  \eqref{nr21} and \eqref{nr31} and  the equality of
\eqref{nr31} and \eqref{nr41} is guaranteed by Theorem \ref{MinMax}.
Finally let us justify the equality between  \eqref{nr41} and  \eqref{nr51}: indeed if $\pi$ is not a martingale measure, then 
by Lemma \ref{DifferentChar} for some $j$ there is a function $\Delta_j$ such that  
$$B=\int \Delta_j(x_1,\dotsc,x_j)(x_{j+1}-x_j)\, d\pi(x_1,\dotsc,x_n)$$ does not vanish. By appropriately scaling $\Delta$ the value of $B$ can be made arbitrarily large.

\medskip

Next assume that $\Phi:\R^n\to [0,\infty]$ is merely lower
semi-continuous and  pick a sequence of bounded continuous functions
$\Phi_1\leq \Phi_2\leq\ldots$ such that $\Phi= \sup_{k\geq 0}
\Phi_k$.  In the following paragraph we will write $P(\Phi),
D(\Phi), P(\Phi_k)$, resp.\ $D(\Phi_k)$ to emphasize the dependence
on the cost function. For each $k$ pick $\m _k\in \Pi (\mu_1,\ldots,
\mu_n)$ such that
$$ P (\Phi_k)\geq \int \Phi\, d\m _k-1/k.$$
Passing to a subsequence if necessary, we may assume that $(\m _k)$
converges weakly to some  $\m \in \Pi (\mu_1,\ldots, \mu_n)$. Then
\begin{align}\begin{split}
P (\Phi)\leq  \int \Phi\, d\m  =\lim_{m\to \infty}\int \Phi_m \, d\m &=\lim_{m\to \infty}\left(\lim_{k\to \infty}\int \Phi_m \, d\m _k\right)\\
&\leq \lim_{m\to \infty} \left(\lim_{k\to \infty}\int \Phi_k \, d\m
_k\right)=\lim_{k\to\infty}P  (\Phi_k).
\end{split}\end{align}
Since $ P  {(\Phi_k)}\leq P  (\Phi)$ it follows that $D  (\Phi)\geq
D  (\Phi_k)=P  ({\Phi_k})\uparrow P  (\Phi)$.

 \medskip

It remains to prove that the optimal value of the primal problem is
attained. To establish this, we use the lower semi-continuity of
$\int \Phi\, d\pi$ on $\Pi(\mu_1, \ldots, \mu_n)$:
 if  a sequence of measures $(\pi_k)$ in $\Pi(\mu_1, \ldots, \mu_n)$ converges weakly to a measure $\pi$, then
\begin{align}\liminf_{k\to\infty} \int \Phi\,d\pi_k\geq \int \Phi\, d\pi.\label{LSCC}\end{align}
 We refer the reader to \cite[Lemma 4.3]{Vi09} for a proof of this assertion.

If $P=\infty$, the infimum is trivially attained, so assume
$P<\infty$ and pick a sequence $(\Q_k)$ in $\M(\mu_1, \ldots,
\mu_n)$ such that $P=\lim_k\int \Phi\,d\Q_k$. As $\M(\mu_1, \ldots,
\mu_n) $ is compact, $(\Q_k)$ converges to some measure $\Q$ along a
subsequence and $\Q$ is a primal minimizer  by \eqref{LSCC}.
 \end{proof}

As we have just seen, the existence of a primal optimizer $\Q$ is basically a consequence of the compactness of the set of all martingale transport plans. The dual set of sub-hedges does not exhibit nice compactness properties and as we already mentioned the dual supremum is not necessarily attained (Proposition \ref{DualEx} below). Although we are not able to give a positive criterion in this direction, it seems worthwhile to comment on the consequences of attainment of the dual problem.

 Assume that there exists a dual maximizer, i.e.\ that there exist $\mu_i$ integrable functions $u_i$ and continuous bounded functions $\Delta_j $ such that the corresponding subhedge  
 (cf.\ \eqref{DualCandidate}) satisfies 
\begin{align}\label{IneqPart} \Psi_{(u_i), (\Delta_j)}\leq \Phi\end{align} and $$ \sum_{i=1}^n \E_{\mu_i} [u_i]=P.$$ Let $\Q$ be a primal optimizer, i.e.\ a martingale measure satisfying the given marginal constraints as well as $\E_\Q[\Phi]= P$. 
Then we have 
$$ 0\leq \E_\Q[ \Phi-\Psi_{(u_i), (\Delta_j)}] = P - D=0.$$ 
As a consequence,  equality holds $\Q$-a.s.\ in \eqref{IneqPart}.
The financial interpretation is that under the market model $\Q$, the payoff $\Phi$ is perfectly replicated through the semi-static hedge corresponding to $(u_i), (\Delta_j)$.

\section{Further analysis in the two dimensional case.}\label{sec:further_analysis}
Throughout this section we focus on the two-period case, i.e.\ $n=2$. 
We start with two examples which illustrate (the general) Theorem \ref{MainTheorem}. Then we show that the dual supremum is not necessarily attained. Finally we explain a conjugacy relation which is relevant for the dual problem and resembles  a well-known concept from the classical theory of optimal transport. 

\medskip

\subsection{A numerical example: forward-start options.}\label{sec:further_analysis:subsec:numerical_example}

We consider the problem to find optimal upper and lower bounds for forward-start options with payoffs $$\Phi_K(s_1,s_2)=(s_2 - K s_1)^+, \quad K=0.5, \ldots, 1.5.$$ Recently, Hobson and Neuberger \cite{HoNe12} have obtained 
 interesting results on model-independent bounds for the  forward-start straddle $|s_2-s_1|$. Since  $|s_2 - s_1|=2(s_2 - s_1)^+-(s_2-s_1)$, this is equivalent  to the case $K=1$, $\Phi_1(s_1,s_2)= (s_2 - s_1)^+ $. An unfortunate feature is that no fully explicit solution is known for generic measures $\mu_1$ and $\mu_2$. In \cite[Section 9]{HoNe12}  numerical upper bounds are obtained in the  cases where $\mu_1, \mu_2$ are given as uniform resp.\ log-normal distributions.
 
We will consider the cases of different strikes and laws $\mu_1, \mu_2$ inferred from market data. 
 By using a linear programming  algorithm, we have computed numerically the optimal lower and upper bounds for different values of $K$. 
 
 The measures $\mu_1$ and
$\mu_2$ are deduced from the prices of call options written on the DAX 
 (pricing date = 2nd Feb.\ 2012) with $t_1 =  1$ year and $t_2= 1.5$ years with $m=18$ strikes ranging from   $30\% $ to $200\%$ of the current spot price $s_0$. The dual for the upper bound reads as (setting $K_{1,0} = K_{2,0}  = 0$)
\begin{align}
D=&\ \inf_{b,c_{i,l},\Delta} b +\sum_{i=1}^2 \sum_{l=0}^{m}c_{i,l}  \mathcal C(t_i,K_{i,l})   \\
  \mbox{s.t. }&\notag \\
F(s_1,s_2):=&\ b +\sum_{i=1}^2 \sum_{l=0}^{m}c_{i,l} (s_i-K_{i,l})_+  + \Delta(s_1)(s_{2}-s_{1})   \geq (s_2- K s_1)^+, \ (s_1,s_2)\in \R_+^2.
\end{align}
 The additional term $\Delta_0(s_0)(s_1 - s_0)$ has been incorporated by considering a vanilla option at $t_1$ with a zero strike. 
 Note that the  function $s_2\mapsto F(s_1,s_2)-(s_2- K s_1)^+$ is piecewise linear with respect to  $s_2$ and therefore attains its extremal values at the points $s_2 = \{ K_{2,j}\}_{j=1,\ldots,m}$, $s_2 = 0$, $s_2 = \infty$, $s_2 = Ks_1$. The above constraints  therefore reduce to $m+3$ constraints parametrized by $s_1$. As a consequence this low-dimensional  linear program can be efficiently implemented by using a classical simplex algorithm \cite{PrTeVeFl07} and by discretizing the spot value $s_1$ on a space grid. We have compared the upper and lower bounds against the prices produced by  models commonly used by practitioners (see Fig.\ \ref{Fig1}): the local volatility model (in short LV) \cite{Du94}, Bergomi's model \cite{Be05} which is a  two-factor variance curve model and finally the local Bergomi model \cite{He09} which has the property to be perfectly calibrated to vanilla smiles at $t_1$ and $t_2$. The LV and local Bergomi models have been calibrated to the DAX implied volatility market. The 
Bergomi model has been calibrated to the variance-swap term structure. As expected, the prices as produced by the LV and local Bergomi models -- consistent with the marginals $\mu_1$ and $\mu_2$ -- are within  our bounds.\footnote{We would like to emphasize that the lower/upper bounds corresponding to different strikes $K$ are attained by different martingale measures. This is not the case if we 
 do not include the martingality constraint as in this case the upper/lower bounds are attained by the co-monotone resp.\ anti-monotone coupling for each strike $K$  (see for instance \cite[Section 2.2.2]{Vi03}). }
 
\begin{figure}[h] 
 \caption{Lower/Upper bounds versus (local) Bergomi and LV models for forward-start options (quoted in Black-Scholes volatility $\times 100$). Parameters for the Bergomi model: $\sigma=2$, $k_1=4$, $k_2=0.125$, $\rho=34.55\%$, $\rho_\mathrm{SX}=-76.84\%$, $\rho_\mathrm{SY}=-86.40\%$. 
 As the Bergomi model is not calibrated to the vanilla smiles, it  may yield implied volatilities below the lower bound, cf.\ Strike $K=1.5$. Notice also that for $K=1.5$ the implied volatility of the Bergomi and LV  model coincides with the lower bound to the level of numerical accuracy.}.\label{Fig1}
 \centering
   \includegraphics[scale=0.3]{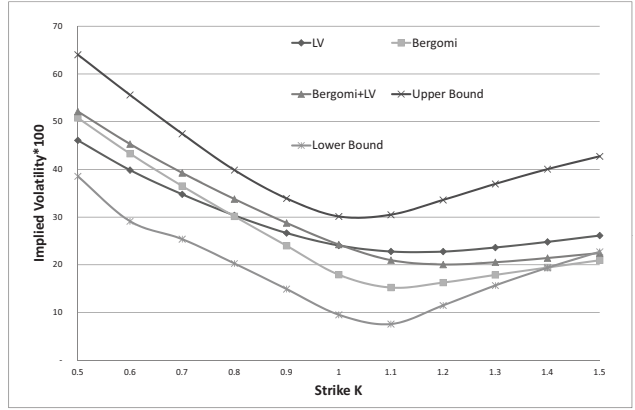}
\end{figure}

\no We have also plotted $F(s_1,s_2)$ as a function of $s_1$ and $s_2$ for the at-the-money forward-start option (i.e. $K=1$, see Fig. \ref{Fig2}) to check the super-replication strategy.  

 \begin{figure}[h]
 \caption{Super-replication strategy for $K=1$: $F(s_1,s_2)$ as a function of ${s_1 \over s_0}$ and ${s_2 \over s_0}$.} \label{Fig2}
   \centering
   \includegraphics[scale=0.3,angle=90]{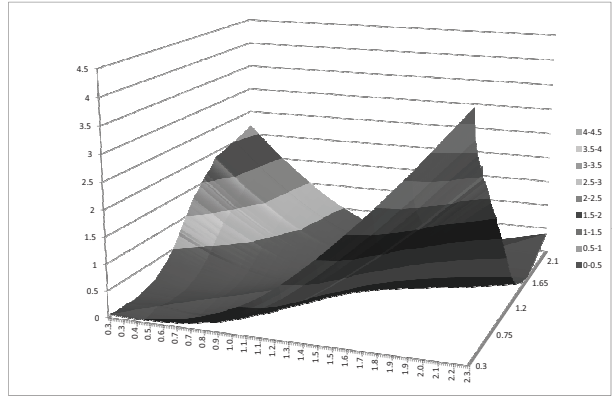}
\end{figure}

 Our result shows that forward-start options are poorly constrained by vanilla smiles. As a conclusion, the practice in the old-quant community to calibrate stochastic volatility models on vanilla smiles to price exotic options (depending strongly on forward volatility) is inappropriate.

 Additional numerical examples are investigated in a companion paper \cite{He11a}. In the case of Asian options the bounds are tighter, indicating that this option can be fairly well hedged with vanilla options. 
 
  We would like to highlight that for  general exotic options, our dual bound can be framed into a large-scale semi-infinite linear program whose numerical implementation requires advanced simplex algorithm such as a primal-dual algorithm within a cutting-plane algorithm \cite{He11a}.

\subsection{Analysis of a theoretical example}\label{sec:further_analysis:subsec:theoretical_example}

We consider a forward-start straddle with payoff function $\Phi (S_1,S_2)=|S_2 - S_1|$; as above we assume that the marginal laws $\mu_1, \mu_2$ are  fixed. 
As mentioned before, Hobson and Neuberger \cite{HoNe12} treat the problem to find a market model which \emph{maximizes} the price $\E_\Q [\Phi (S_1,S_2)]$; specific examples are worked out in detail.

Here we focus on the problem to  minimize $\E_\Q [\Phi (S_1,S_2)]$ in a concrete example.
The marginals $\mu_1,\mu_2$ are defined by the respective densities (where we write $\lambda$ for the Lebesgue measure)
$$\dm(s_1)=\frac12\1_{[-1,1]},\quad  \dn(s_1) = \frac{2+s_1}{3}\1_{[-2,-1]} + \frac{1}{3} \1_{[-1,1]} + \frac{2-s_1}{3}\1_{[1,2]},$$
cf.\ Figure \ref{fig_support_marginals} below.
Recall that the primal, resp.\ dual problem is then given by
   \begin{align*}
    P =&  \inf_{\Q\in \M(\mu_1,\mu_2)} \E_{\mathbb{Q}}[|S_2 - S_1|], \\
    D =& \sup_{u_1, u_2: \exists \Delta, u_1(s_1)+u_2(s_2)+\Delta (s_1)(s_2-s_1) \leq |s_2-s_1|} \E_{\mu_1}[u_1] + \E_{\mu_2}[u_2].
   \end{align*}
By Theorem \ref{MainTheorem} we know that there is no duality gap, i.e.\ $P=D$. Our aim is to determine the primal minimizer $\Q$ as well as dual maximizers $u_1, u_2, \Delta$. We follow the common procedure of guessing and verification: i.e.\ making various (unjustified) assumptions we will first produce explicit candidates. Then it is possible to verify rigorously that these candidates indeed solve the given problem.

 Due to Hobson   \cite{Ho12}  (see also (\cite[Section 6]{BeJu12}) one expects that the primal minimizer $\Q$  has a very particular structure:  
 Writing  $(\Q_{s_1})_{s_1\in [-1,1]}$ for the disintegration\footnote{In probabilistic terms, the measure $\Q_{s_1}$ is the conditional distribution of $S_2$ under $\Q$ given that $S_1={s_1}$.} of $\Q$ w.r.t.\ $\mu_1$, each measure $\Q_{s_1}$ will be concentrated on three points. More specifically we \emph{guess}\footnote{We emphasize that while this simple guess works in the present setting, the situation is more subtle for general distributions. } that there exist monotone decreasing functions $f\colon[-1,1] \to [-2,-1], g\colon[-1,1]\to [1,2]$ such that  $\textrm{supp}(\Q_{{s_1}}) = \{f({s_1}),{s_1},g({s_1})\}$.

Figure \ref{fig_support_marginals} depicts the measures $\mu_1, \mu_2$ and for each particle starting in ${s_1}\in [-1,1]$ the possible positions $f({s_1}),{s_1},g({s_1})$ at time $t=2$. 

\begin{figure}[ht]
  \caption{Marginals and Support of Primal Optimizer}\label{fig_support_marginals}
  \centering
    \includegraphics{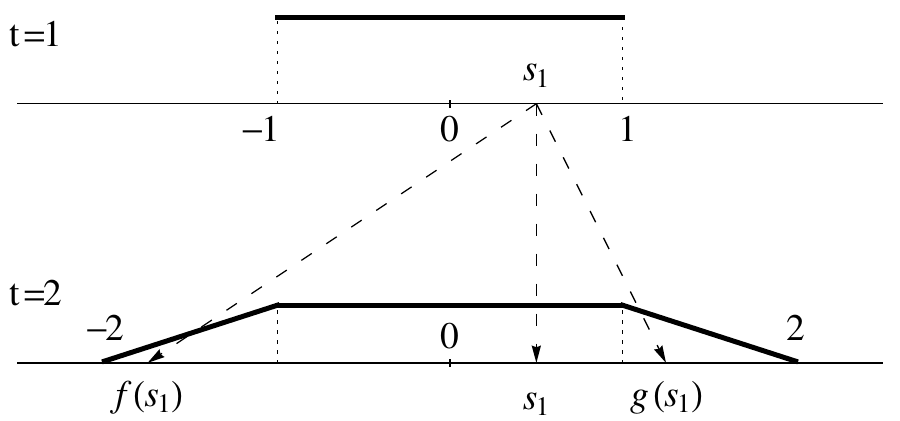}
\end{figure}

I.e., as much mass as possible remains at its place, the rest is either moved to the interval on the left of $[-1,1]$ (via $f$), or to the right (via $g$).
For ${s_1}\in [-1,1]$, we write the measure $\Q_{s_1}$ as 
$$\Q_{s_1} = a({s_1}) \delta_{f({s_1})}+ b({s_1})\delta_{s_1}+ c(s_1) \delta_{g(s_1)}, \quad \mbox{where $a(s_1)+b(s_1)+c(s_1)=1$}.$$
Taking for granted that $f, g$ are sufficiently smooth, the marginal conditions on $\Q$ translate to 
\begin{align*}
  \dm(s_1) a(s_1) &= (-f'(s_1)) \dn(f(s_1)), \\
  \dm (s_1)b(s_1) &=  \dn(s_1)  \\ 
  \dm(s_1) c(s_1) &= (-g'(s_1))\dn(g(s_1)).
\end{align*}

Thus $b(s_1)= 2/3$ and $a(s_1), c(s_1)$ can be expressed in terms of the functions $f,g$, i.e. $$a(s_1)=-f'(s_1) \dn(f(s_1)) / \dm(s_1),\qquad c(s_1)= - g'(s_1) \dn(g(s_1)) / \dm(s_1).$$ From $a(s_1)+ b(s_1)+ c(s_1) =1$, we obtain for $s_1\in [-1,1]$ the equation 
\begin{align}
\dn(f(s_1)) (-f'(s_1)) + \dn(g(s_1)) (-g'(s_1)) &=  \frac 13 \dm(s_1)= \frac 16   .			\label{dgl_supp_mass}
\end{align}
The martingale property is expressed by $f(s_1) a(s_1)+ s_1 b(s_1) + g(s_1) c(s_1)= s_1$. In terms of $f,g$ this amounts to 
   \begin{align}
      \dn(f(s_1)) f(s_1) (-f'(s_1)) + \dn(g(s_1))g(s_1) (-g'(s_1)) &= \frac {s_1}3\dm({s_1}) =\frac {s_1}6.		\label{dgl_supp_mart} 
   \end{align}

Adding the initial conditions $f(1) = -2$ and $g(1)=1$, the differential equations \eqref{dgl_supp_mass}, \eqref{dgl_supp_mart} have the unique solution 
   \begin{equation}
      f({s_1}) = -(3+{s_1})/2, \qquad  g({s_1}) = (3-{s_1})/2.
   \end{equation}
These functions $f,g$ determine a martingale measure $\Q\in \M(\mu_1, \mu_2)$ which is our candidate optimizer for the primal problem.

A dual optimizer $(u_1,u_2, \Delta)$ consists of functions $u_1, u_2, \Delta:\R\to\R$ satisfying
   \begin{equation}
      u_1({s_1}) + \Delta({s_1})({s_2}-{s_1}) \le |s_2-{s_1}| - u_2({s_2}) \label{dual_constraints_ex}
   \end{equation}
 for all $({s_1},{s_2}) \in \R^2$.  
From the considerations at the end of Section \ref{sec:proof_of_main_thm} we know that equality should hold in \eqref{dual_constraints_ex} for all $({s_1},{s_2})$ in the support of the primal minimizer. Since we anticipate that $\Q$ is this minimizer, we expect equality in \eqref{dual_constraints_ex} for ${s_2} \in \{f({s_1}),{s_1},g({s_1})\}, {s_1}\in [-1,1]$. 
Writing $u_{2, l}:=u_2\rvert_{(-\infty,-1]}$, $u_{2,m}:=u_2\rvert_{[-1,1]}$ and \linebreak \mbox{$u_{2, r}:=u_2\rvert_{[1,\infty)}$}, this amounts to 
\begin{align}\label{FirstSys}\begin{split}
u_1({s_1}) + \Delta({s_1})(f({s_1})-{s_1}) &= |f({s_1})-{s_1}| - u_{2,l}(f({s_1})),\\
u_1({s_1}) + \Delta({s_1})({s_1}-{s_1}) &= |{s_1}-{s_1}| - u_{2,m}(s_1)\quad\quad \Longleftrightarrow \ u_1({s_1})= -u_{2,m}({s_1})  \\
u_1({s_1}) + \Delta({s_1})(g({s_1})-{s_1}) &= |g({s_1})-{s_1}| - u_{2,r}(g({s_1})),
\end{split}
\end{align}
for $s_1\in[-1,1]$. 
Furthermore it is reasonable to assume that for fixed ${s_1}\in [-1,1]$, the affine function $ s_2 \mapsto  u_1({s_1}) + \Delta({s_1})({s_2}-{s_1})$ is tangent to the function $ s_2 \mapsto |{s_2}-{s_1}| - u_2({s_2}) $ if ${s_2}$ equals $f({s_1})$ resp.\ $g({s_1})$. 
This leads us to identify the slope $\Delta({s_1})$ of this affine function with the derivatives of the right hand side for ${s_2}\in \{f({s_1}), g({s_1})\}$
\begin{align}\label{SecondSys}
\begin{split}
\Delta({s_1}) &= \partial_{s_2} \Big(({s_1}-{s_2}) - u_{2,l}({s_2})\Big)\big\rvert_{{s_2}=f({s_1})} = -1-u_{2,r}'(f({s_1}))\\
\Delta({s_1}) &= \partial_{s_2} \Big(({s_2}-{s_1}) - u_{2,r}({s_2})\Big)\big\rvert_{{s_2}=g({s_1})} = 1 - u_{2,l}'(g({s_1})).
\end{split}
\end{align}
The equations \eqref{FirstSys} resp.\ \eqref{SecondSys} are solved by 
\begin{align*}
&u_1({s_1}) = (9-5{s_1}^2)/6=-u_{2,m}({s_1}), \quad \Delta({s_1}) = -2 {s_1}/3.\\
 &u_{2,l}({s_2}) = -3 -3{s_2} -2s_2^2 /3, \quad  u_{2,r}({s_2}) = -3 + 3{s_2} -2 s_2^2 /3.
\end{align*}
Setting $u_2 = u_{2,l} \1_{(-\infty,\-1]} + u_{2,m} \1_{[-1,1]} + u_{2,r} \1_{[1,\infty)}$, we have thus found a ``reasonable'' candidate solution for the dual problem. 
It is then straightforward  to verify that $(u_1, u_2, \Delta)$ is admissible, i.e., satisfies \eqref{dual_constraints_ex}.

\medskip 

To verify that $\Q$ resp.\ $(u_1, u_2, \Delta)$ are in fact solutions of the primal resp.\ dual problem we evaluate the corresponding functionals
\begin{align*}
 & \E_{\Q}[|S_2 - S_1|] = \int_{} \int_{} |{s_2}-{s_1}| \, d\Q_{s_1}({s_2}) \, d\mu_1({s_1}) 
= \frac{1}{3}, \\
 & \E_{\mu_1}[u_1] + \E_{\mu_2}[u_2]
   = \int_{-1}^1 u_1({s_1}) \dm({s_1})\, d{s_1} + \int_{-2}^{2} u_{2}({s_2})\dn({s_2}) \, d{s_2} 
   = \frac{1}{3}.
\end{align*}
Hence $P= \frac13 = D$ and we conclude that $\Q$ resp.\ $(u_1, u_2, \Delta)$ are indeed the desired solutions.

\subsection{Non-Existence of dual maximizers}\label{sec:further_analysis:subsec:counterexample}

In the classical optimal transport problem, the optimal value of the dual problem is
attained provided that the cost function is bounded (\cite[Theorem
2.14]{Ke84}) or satisfies appropriate moment conditions
(\cite[Therorem 2.3]{AmPr03}).

This is not the case in our present setting as Proposition \ref{DualEx} shows that  the dual
supremum \eqref{Dual} is not necessarily attained even if $\mu_1,\mu_2$ are compactly supported. 
Our counterexample fits into the framework\footnote{Formally Hobson and Neuberger are interested to \emph{maximize} the payoff of $|S_2-S_1|$ while we are interested to \emph{minimize} the payoff $-|S_2-S_1|$. Mathematically, the two problems are of course the same. We haven chosen the latter formulation to be consistent with the notation in our main result Theorem \ref{MainTheorem}.} of \cite{HoNe12}, i.e.\ we consider two periods and an exotic option with payoff $ -|S_2-S_1|$.

\begin{prop}\label{DualEx}
Let $\mu_2$ be the uniform distribution on the interval $[0,2]$ and $$\Phi({s_1},{s_2})=-|{s_2}-{s_1}|.$$ There exists a measure $\mu_1$, concentrated on countably many atoms, such that the (finite) dual value is not attained.

Moreover, there do not exist functions $u_1, u_2, \Delta:\R\to\R$ such that  
\begin{align}
\label{EqualityPart}
\begin{split}
u_1({s_1})+u_2({s_2})+\Delta({s_1})({s_2}-{s_1})&\leq -|{s_2}-{s_1}|, \quad \mbox{for all  $({s_1},{s_2})\in \R^2$},\\
u_1({s_1})+u_2({s_2})+\Delta({s_1})({s_2}-{s_1})&= -|{s_2}-{s_1}|, \quad \mbox{for $\Q$-a.a.\ $({s_1},{s_2})\in \R^2$,} 
\end{split}
\end{align}
where $\Q$ is a minimizer of the primal problem.
\end{prop}
In the proof of Proposition \ref{DualEx} we will use the following auxiliary result. 

\begin{lem}\label{SepProblems}
Assume that $\mu_1, \mu_2$ are probability measures on $\R$ having finite first moments, let $\Q\in \M(\mu_1,\mu_2)$ and fix $ s\in\R$. 
The following are equivalent.
\begin{enumerate}[(i)]
\item The call prices $\E_{\Q}[(S_1-s)^+]=\int ({s_1}-s)^+\, d\mu_1({s_1})$ and $\E_{\Q}[(S_2-s)^+]=\int ({s_2}-s)^+\, d\mu_2({s_2})$ are equal.
\item If $S_1 \leq s$, then $S_2\leq s$ and if $S_1>s$ then $S_2\geq s$, $\Q$-a.s.
\end{enumerate}
In particular, if (ii) holds for one measure in $\M(\mu_1, \mu_2)$, then it applies to all elements of $\M(\mu_1, \mu_2)$.
\end{lem}
\begin{proof}
Given a random variable $X$ and a measurable set $A$ we write $\E_{\Q}[X,A]= \E_\Q[X\1_A]$.
Then we have 
\begin{align}
\label{SHalf}\E_{\Q}[(S_2-s)^+,\ S_1> s]&\geq \E_{\Q}[S_2-s,\ S_1> s],
\\
\label{FHalf}\E_{\Q}[(S_2-s)^+,\ S_1\leq s]&\geq 0, 
\end{align}
where equality holds in \eqref{SHalf}
if and only if $S_1> s \Rightarrow S_2\geq s$ $\Q$-a.s. and in \eqref{FHalf} if and only if $S_1\leq s \Rightarrow S_2\leq s$  $\Q$-a.s. 
Using (in deriving the last line) that $S$ is a $\Q$-martingale we thus obtain
\begin{align*}
\E_{\Q}[(S_2-s)^+]&=\E_{\Q}[(S_2-s)^+,\ S_1>s]+\E_{\Q}[(S_2-s)^+,\ S_1\leq s]\\
& \geq  \E_{\Q}[S_2-s,\ S_1>s] +0 \\
& = \E_{\Q}[S_1-s,\ S_1>s]+  \E_{\Q}[(S_1-s)^+,\ S_1\leq s]=\E_\Q[(S_1-s)^+],
\end{align*}
with equality holding true if and only if (ii) is satisfied. 
 \end{proof}
We also make the following trivial observation:
\begin{lem}\label{ReallySimple}
Let $c,d, x\in \R, c<x \leq d$, let $m$ be a measure on $[c,d]$ and set $\alpha=m([c,d])$. Then the product-measure $\delta_x \otimes m$ is the unique  measure on $[c,d]^2$ which has $\alpha \delta_x$ as first marginal and $m$ as second marginal.
\end{lem}

\begin{proof}[of Proposition \ref{DualEx}.]
Denote by $\lambda$ the Lebesgue measure on the real line and
set $\mu_2=\frac{1}{2} \lambda_{[0,2]}$.   
Define 
\begin{align}
\textstyle a_n&= \textstyle \tfrac 12 \Big( \sum_{i=1}^{n-1} \tfrac1{i^2}+\sum_{i=1}^{n} \tfrac1{i^2}\Big), n\geq 1\\
\bar a &= \textstyle  \tfrac 12 \Big(\sum_{i=1}^{\infty} \tfrac1{i^2}+2\Big)= \tfrac{\pi^2}{12}+1,\\
\mu_1&=\tfrac{1}{2}\sum_{i=1}^\infty \tfrac1{i^2} \delta_{a_i}+\tfrac12 
\Big(2- \tfrac{\pi^2}6\Big) \delta_{\bar a}. 
\end{align} 
We claim that $\M(\mu_1, \mu_2)$ consists of the single element
    \begin{align}
    \Q= \tfrac{1}{2}\sum_{n=1}^\infty  \delta_{a_n}\otimes \lambda_{\left|\left[\sum_{i=1}^{n-1} \frac1{i^2},\sum_{i=1}^{n} \frac1{i^2}\right]\right.}+\tfrac12 \delta_{\bar a}\otimes \lambda_{\left|\left[\frac{\pi^2}6,2\right]\right.}.
    \end{align}
Note that $a_n$ is defined to be the midpoint of the interval 
$I_n:=\left[\sum_{i=1}^{n-1} \frac1{i^2},\sum_{i=1}^{n} \frac1{i^2}\right]$, $n\in\N$; likewise $\bar a$ is the midpoint of $\bar I = \left[\frac{\pi^2}6,2\right]$.
 Therefore  $\Q$ is indeed an element of $\M(\mu_1,\mu_2)$. 
 
To prove that $\Q$ is the only element of $\M(\mu_1,\mu_2)$
we first observe that for $s\in S =\{\sum_{i=1}^n \tfrac 1{i^2}:n\geq 0\}\cup\{\tfrac{\pi^2}6,2\}$
    \begin{align}\label{TransBoxes}
    \Q\big([0, s]^2 \cup [s,2]^2\big)=1.
    \end{align}
 Lemma \ref{SepProblems} yields that \eqref{TransBoxes} applies to an arbitrary measure $\widetilde \Q\in \M(\mu_1,\mu_2)$. As $S$ is countable, it follows that 
$$ 1 =\widetilde \Q\left(\bigcap_{s\in S}\Big( [0, s]^2 \cup [s,2]^2\Big)\right).$$
We also note that  
$$
\bigcap_{s\in S}\Big( [0, s]^2 \cup [s,2]^2\Big)
= \bigcup_{n=1}^\infty I_n^2
\cup \bar I^2
 =:\Gamma.
$$
Applying Lemma \ref{ReallySimple} with $[c,d]=I_n,$ 
$ n\in\N$ resp.\ $[c,d]= \bar I$ 
it follows that $\Q$ is the only measure satisfying $\Q(\Gamma)=1$ and having marginals $\mu_1, \mu_2$. Since $\widetilde \Q(\Gamma)=1$, we conclude that $\widetilde \Q=\Q$. Thus we have  indeed $\M(\mu_1, \mu_2)=\{\Q\}$.

\begin{figure}[h]
  \caption{Support of the unique Martingale Measure}\label{fig_support_counterexample}
  \centering\vspace*{0.2cm}
    \includegraphics[scale=0.2]{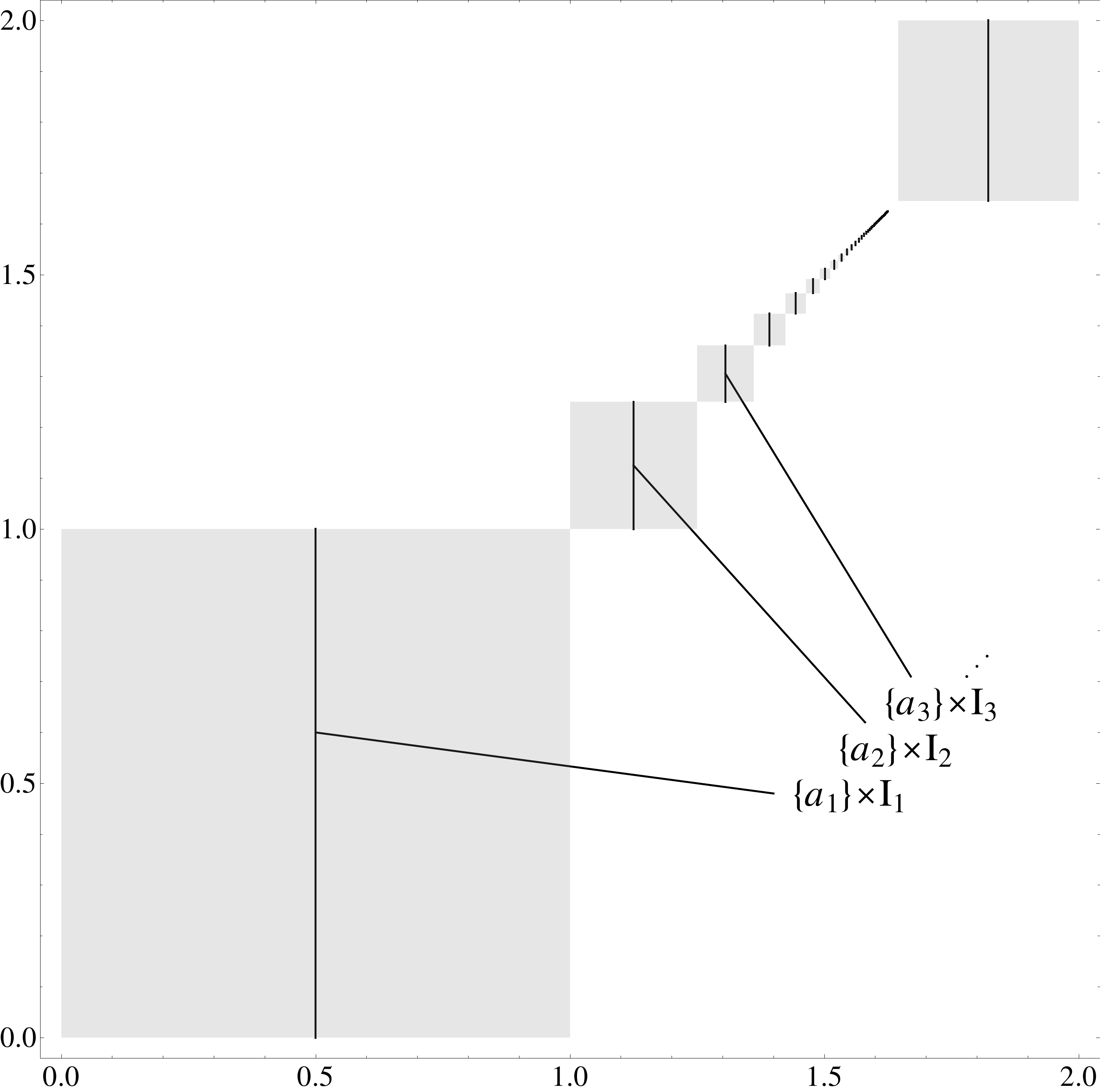}
\end{figure}


According to the short discussion preceding Proposition \ref{DualEx} it is sufficient to show that \eqref{EqualityPart} cannot be verified.
Striving for a contradiction, we assume that there exist
$u_1,u_2,\Delta\colon\R\to \R $ such that \eqref{EqualityPart}  (with respect to the  measure $\Q$) holds true.

Setting $d_n=u_1(a_n), k_n:= \Delta (a_n) $, $n\in\N$ we obtain
\begin{align}\label{ExStart} d_n+ k_n({s_2}-a_n) +|{s_2}-a_n|\leq -u_2({s_2})\end{align}
for ${s_2}\in\R$ with equality holding for $\lambda$-almost all $s_2\in  I_n$. 
Applying this with $n$ and $n+1$, respectively yields  
      \begin{align}\label{1slope}
     \begin{split} 
    \textstyle d_n+ k_n({s_2}-a_n) +|{s_2}-a_n|\leq  -u_2(s_2)=\ &d_{n+1}+ k_{n+1} ({s_2}-a_{n+1}) +|{s_2}-a_{n+1}| \\
	    &\quad \mbox{for ${s_2}\in I_{n+1}$,}  \\  
      d_n+ k_n({s_2}-a_n) +|{s_2}-a_n| =  -u_2(s_2)\geq\  & d_{n+1}+ k_{n+1} ({s_2}-a_{n+1}) +|{s_2}-a_{n+1}| \\
	   &\quad  \mbox{for ${s_2}\in I_n$}.
	    \end{split}
      \end{align}
Note that as these inequalities appeal to piecewise linear functions it is not necessary to exclude  exceptional null-sets, in particular
$$d_n+ k_n(y_0-a_n) +|y_0-a_n|= d_{n+1}+ k_{n+1} (y_0-a_{n+1}) +|y_0-a_{n+1}| $$
for $y_0=\sum_{i=1}^{n} \tfrac1{i^2}.$
It follows that the slope of ${s_2}\mapsto d_n+ k_n({s_2}-a_n) +|{s_2}-a_n|$ is smaller or equal than the one of ${s_2} \mapsto d_{n+1}+ k_{n+1} ({s_2}-a_{n+1}) +|{s_2}-a_{n+1}|$ at the point ${s_2}=y_0$,
i.e.
      \begin{align}
      k_n+1\leq k_{n+1}-1. \label{recursion_kn} 
      \end{align}
Hence $k_{n}\geq (k_1-2)+ 2 n.$ 

Applying \eqref{1slope} for ${s_2}= a_{n+1}$ we obtain
      \begin{align}
      & d_n+ k_n(a_{n+1}-a_n) +|a_{n+1}-a_n|\leq d_{n+1}, \label{recursion_dn}\\
      \Longrightarrow\quad & d_n+ k_n  \tfrac 12 (\tfrac1{n^2}+\tfrac1{(n+1)^2}) \leq d_{n+1}. \label{AlmostThere}
      \end{align}
Iterating \eqref{AlmostThere}, we arrive at 
      \begin{align*}d_{n+1}&\geq d_1 + \sum_{i=1}^n  [ (k_1-2)+ 2 i]  \tfrac 12(\tfrac1{i^2}+\tfrac1{(i+1)^2})\\
      &\geq d_1 - |k_1-2|  \sum_{i=1}^n  \tfrac 12 (\tfrac1{i^2}+\tfrac1{(i+1)^2})  + \sum_{i=1}^n i\, (\tfrac1{i^2}+\tfrac1{(i+1)^2})
      \geq  d_1 -|k_1 -2| \tfrac {\pi^2}6 + \sum_{i=1}^n \tfrac1i .
      \end{align*} 
Thus, $d_n$ and $ k_n$ tend to $\infty $ as $n$ goes to infinity. Combining this with \eqref{ExStart}, it follows that $-u_2({s_2})=\infty $ for ${s_2}\geq \tfrac{\pi^2}6$.
 \end{proof}

Arguably, the counterexample obtained in Proposition \eqref{DualEx} is rather artificial. In particular a crucial property is that the problem consists of infinitely many problems which are mutually not connected: there exist infinitely many intervals  which intersect only in boundary points such that every $\Q\in \M(\mu_1,\mu_2)$ is concentrated on the union of the squares-products of these intervals. By Lemma \ref{SepProblems}
 this is reflected in the prices of European calls by the property 
$$\E_\Q[(S_1-s)^+]=\E_\Q[(S_2-s)^+]$$
whenever the strike $s$ is the endpoint of some interval.

   Clearly it would be desirable to find conditions which guarantee that the dual supremum is attained. However we are not able to do so at the present stage.\footnote{Some progress in this direction is made in \cite[Appendix A]{BeJu12}. (Note added in revision.)}

\subsection{A c-convex approach}\label{c-convex-approach}\label{sec:further_analysis:subsec:c-convex_approach}

In the dual part of the usual transport problem  it suffices to maximize over all pairs of functions $(u_1, u_2)$ where $u_1$ is the conjugate of $u_2$ with respect to
 $\Phi$, i.e., satisfies $$u_1({s_1})= \inf_{s_2} \Phi({s_1},{s_2})-u_2({s_2}).$$   (We refer the reader to \cite[Section 2.4]{Vi03}, \cite[Chapter 5]{Vi09} for details on this topic.)

An analogous result holds true in the present martingale setup.
Its relevance stems from the fact that it simplifies the construction of hedging strategies for options depending on two future time points. Unfortunately we are not aware of a generalization to the multi-period case.

 Given a function $g:\R\to (-\infty,
\infty]$, we write $g^{**}$ for its convex envelope\footnote{I.e.\
$g^{**}:\R\to\R$ is the largest convex function smaller than or equal to 
$g$.}. For $G\colon\R^2\to \R$, let $G^{**}\colon\R^2\to \R$ be the function
satisfying $$ G^{**}({s_1},.)=\big(G({s_1},.)\big)^{**}$$ for every ${s_1}\in \R$.
(It is straight forward to prove that $G^{**}$ is Borel measurable
resp.\ lower semi-continuous whenever $G$ is.)

\begin{prop}
 Let $\Phi:\R^2\to (-\infty, \infty]$ be a lower semi-continuous function such that   $\Phi({s_1},{s_2})\geq - K(1+|{s_1}|+|{s_2}|), {s_1},{s_2}\in\R$ and assume that there is some $\m\in \M(\mu_1, \mu_2)$ satisfying $\E_{\Q}[ \Phi] <\infty$. Then
\begin{align}\label{NiceDuality}
 P= \sup_{u_2\colon\R\to \R, \int |u_2|\, d\mu_2 <\infty}
  \E_{\mu_1}[(\Phi(S_1,S_1) - u_2(S_1))^{**}] + \E_{\mu_2}[u_2(S_2)].
\end{align}
\end{prop}
(In the course of the proof we will see that for every choice of
$u_2$ the first integral in  \eqref{NiceDuality} is well defined,
assuming possibly the value $-\infty$.)
\begin{proof}
We start to show that the primal value $P$  is greater or equal than
the right hand side of  \eqref{NiceDuality}. Let $u_2\colon\R\to \R$ be a
$\mu_2$-integrable function. For  $\m\in \M(\mu_1, \mu_2)$ satisfying
$\E_{\Q}[ \Phi] <\infty$ we have
\begin{align}
\nonumber \E_{\Q}[\Phi(S_1,S_2)] &= \E_{\Q}[\Phi(S_1,S_2) - u_2(S_2)] + \E_{\mu_2}[u_2(S_2)] \\
  \nonumber &\ge  \E_{\Q}[(\Phi(S_1,S_2) - u_2(S_2))^{**}] + \E_{\mu_2}[u_2(S_2)] \\
   &= \E_{\mu_1}[\E_{\Q}[(\Phi(S_1,S_2) - u_2(S_2))^{**}| \, S_1]] + \E_{\mu_2}[u_2(S_2)]  \label{jen1} \\
   &\ge \E_{\mu_1}[(\Phi(S_1,\E_{\Q}[S_2| \, S_1]) - u_2(\E_{\Q}[S_2| \, S_1]))^{**}] + \E_{\mu_2}[u_2(S_2)] \label{jen2} \\
\nonumber   &=\E_{\mu_1}[(\Phi(S_1,S_1) - u_2(S_1))^{**}] + \E_{\mu_2}[u_2(S_2)],
\end{align}
where the inequality between \eqref{jen1} and \eqref{jen2}
holds due to Jensen's  inequality.
This proves the first inequality.

To establish the reverse inequality, we make a simple observation.
Let   $s_1\in \R$ and $g\colon\R\to \R$ be some function. Suppose that
for $u_1\in\R$ there exists $\Delta\in \R$ such that
$$u_1+ \Delta \cdot({s_2}-{s_1}) \leq g({s_2})$$
for all $s_2\in \R$. Then $u_1\leq g^{**}({s_1}).$

Applying this for $s_1\in \R$ to the function ${s_2}\mapsto g({s_2})=
\Phi({s_1},{s_2})- u_2({s_2})$ we obtain
\begin{align}
&\sup_{u_2}\,
  \E_{\mu_1}[(\Phi(S_1,S_1) - u_2(S_1))^{**}] + \E_{\mu_2}[u_2(S_2)]
\\
 \geq\  & \sup_{u_2}    \sup_{u_1 \,:\, \exists\Delta,   u_1({s_1})+\Delta({s_1})({s_2}-{s_1}) \leq \Phi({s_1},{s_2})-u_2({s_2})}
 \E_{\mu_1}[u_1(S_1)] + \E_{\mu_2}[ u_2(S_2)] \\
 =\ &\sup_{u_1, u_2\, :\, \exists \Delta,\, \Psi_{u_1, u_2, \Delta}\leq \Phi}  \E_{\mu_1}[u_1(S_1)] + \E_{\mu_2}[ u_2(S_2)]  = D = P\, ,
 \end{align}
 where we tacitly assumed that the suprema are taken over $\mu_i$-integrable functions $u_i\colon\R\to \R, i=1,2$ and that $\Delta\colon\R\to \R$ is bounded measurable.
 \end{proof}

\section*{Summary}

This paper focusses on robust  pricing and hedging of exotic options written on one risky asset. 
\medskip

Given call prices at finitely many time points $t_1, \dotsc , t_n$ the set of martingale models calibrated to these prices leads to an interval of consistent prices of a pre-specified exotic option.  Theorem 1 resp.\ Corollary 1.1 assert inter alia that every price outside this interval gives rise to a \emph{model-independent} arbitrage opportunity. This arbitrage can be realized through a semi-static sub/super-hedging strategy consisting in dynamic trading in the underlying and a static portfolio of call  options.

\medskip

Our approach to these results is based on the duality theory of mass transport.

\section*{Acknowledgements}
We thank  the associate editor and the extraordinarily careful referees for their comments  and in particular for pointing out a mistake in an earlier version of this article. We also benefitted from remarks   by Johannes Muhle-Karbe.

\section*{Appendix}\label{sec:appendix}
As a special case of \cite[Theorem 2.14]{Ke84} we have the duality
equation
\begin{align*}
P_{MK}(\Phi) &= inf\{I_{\pi}(\Phi): \pi\in \Pi(\mu_1, \ldots, \mu_n)\} \\
	     &= \sup\Big\{ \sum_{i=1}^n \int u_i\, d\mu_i:u_1\oplus\ldots\oplus
u_n\leq \Phi, \, u_i \mbox{ is $\mu_i$-integrable}\Big\}
\end{align*}
for every lower semi-continuous cost function $\Phi:\R^n\to
[0,\infty]$. The main task in the subsequent proof of Proposition
\ref{MKDuality} is to show that the duality equation is obtained if one
restricts to functions in the class $\SS$ in the dual problem.
\begin{proof}[of Proposition \ref{MKDuality}]
As in the proof of Theorem \ref{MainTheorem}, it is sufficient to prove the duality equation in the case $\Phi\geq 0$.

\medskip

Given a bounded continuous function $f$ and $\eps >0$, then for
every $i=1,\ldots, n$ there is some $u\in \SS$ such that $f\geq u$
and $\int f-u\, d\mu_i< \eps$.  Therefore we may change the class of
admissible functions from $\SS$ to $C_b(\R)$, i.e.\ it suffices to
prove
\begin{align}\label{CBDuality}
P_{MK}(\Phi) = \sup\Big\{ \sum_{i=1}^n \int u_i\, d\mu_i:
u_1\oplus\ldots\oplus u_n\leq \Phi, \, u_i\in C_b(\R)\Big\}.
\end{align}

\medskip

We will first show this under the additional assumption that $\Phi
\in C_c(\R^n)$. By  \cite[Theorem 2.14]{Ke84} we have that for each
$\eta >0$ there exist $\mu_i$-integrable functions $u_i$,
$i=1,\ldots, n$ such that
 $$P_{MK}(\Phi)-\sum_{i=1}^n \int u_i\, d\mu_i\leq \eta$$ and $ u_1\oplus\ldots\oplus u_n\leq \Phi.$ Note that the latter inequality implies that $u_1, \ldots, u_n$ are uniformly bounded since $\Phi$ is uniformly bounded from above.

To replace $u_1 $ by a function in $C_b$ we consider $H=
\Phi-(u_1\oplus \ldots \oplus u_n)$ and define
\begin{align}\label{BetterU}
\tilde u_1 (x_1):= \inf_{x_2,\ldots, x_n\in \R} H(x_1, \ldots,
x_n)\end{align} for $x_1\in \R$. We claim that $\tilde u_1$ is
(uniformly) continuous. Indeed, as $\Phi$ is uniformly continuous,
for every $\eps>0$ there exists $\delta>0$ such that whenever $x,x'\in
\R$, $|x-x'|<\delta$, then
$$ |H(x,x_2, \ldots, x_n)-H(x',x_2, \ldots, x_n)| = |\Phi(x,x_2, \ldots, x_n)-\Phi(x',x_2, \ldots, x_n)|<\eps.$$
Thus we obtain
$$| \tilde u_1 (x)-\tilde u_1 (x')|=\Big|\inf_{x_2,\ldots, x_n\in \R} H(x,x_2, \ldots, x_n) - \inf_{x_2,\ldots, x_n\in \R} H(x',x_2, \ldots, x_n)\Big|\leq \eps$$
whenever $|x-x'|< \delta$. By definition $\tilde u_1$ is also
bounded from below and satisfies $\tilde u_1\geq u_1$ as well as
$$\tilde u_1 \oplus u_2\oplus \ldots \oplus u_n\leq \Phi.$$
Iteratively replacing the functions $u_2, \ldots, u_{n}$ in the same
fashion, we obtain \eqref{CBDuality} in the case $\Phi\in
C_c(\R^n)$.

 \medskip

 Using precisely the same argument as in the proof of Theorem \ref{MainTheorem},  we obtain the duality relation in the case of a general, lower semi-continuous function $\Phi:\R^n\to [0,\infty]$. 
 \end{proof}

\bibliographystyle{alpha}

\bibliography{joint_biblio}

\end{document}